\documentclass[10pt]{article}

\usepackage[vlined,ruled]{algorithm2e}

\usepackage{mathptmx}

\usepackage{comment}
\usepackage{amsmath}
\usepackage{amssymb}
\usepackage{amsthm}
\makeatletter
\newcommand{\singlespacing}{\let\CS=
\@currsize\renewcommand{\baselinestretch}{1}\tiny\CS}
\newcommand{\foospacing}{\let\CS=
\@currsize\renewcommand{\baselinestretch}{1.9}\tiny\CS}
\makeatother

\usepackage{multirow}

\usepackage{algorithmic}
\newcommand{\pos}{{{{\mathit{pos}}}}}
\newcommand{\p}{\ensuremath{\mathrm{P}}}
\newcommand{\np}{\ensuremath{\mathrm{NP}}}
\newcommand{\npc}{\ensuremath{\mathrm{NP\hbox{-}complete}}}
\newcommand{\npcs}{\ensuremath{\mathrm{NP\hbox{-}comp.}}}
\newcommand{\score}{\ensuremath{\mathit{score}}}
\newcommand{\vetoes}{\ensuremath{\mathit{vetoes}}}

\newtheorem{theorem}{Theorem}[section]

\newtheorem{corollary}[theorem]{Corollary}
\newtheorem{definition}[theorem]{Definition}
\newcommand{\bigoh}{{\protect\mathcal O}}
\newtheorem{tool}[theorem]{Tool}
\newtheorem{lemma}[theorem]{Lemma}

\newcommand{\calS}{{{\mathcal{S}}}}

\newcommand{\SKIPPED}[1]{\footnote{Search the source for \texttt{SKIPPED} to see omitted text.}}

\begin{document}
\sloppy

\title{Weighted Electoral Control}

\author{
  Piotr Faliszewski\\
  AGH University of Science and Technology\\
  Krakow, Poland\\
\and
  Edith Hemaspaandra\\
  Rochester Institute of Technology\\
  Rochester, NY, USA
\and
  Lane A. Hemaspaandra\\
  University of Rochester\\
  Rochester, NY, USA
}

\date{May 4, 2013} 

\maketitle

\begin{abstract}
  Although manipulation and bribery have been extensively studied 
  under 
  weighted voting, there has been almost no work done on
  election control 
  under
  weighted voting.  This is
  unfortunate, since weighted voting appears in many important natural
  settings.  
In this paper, we study the complexity of
  controlling the outcome of weighted elections through adding and deleting
 voters.  %
 We
  obtain polynomial-time algorithms, NP-completeness results, and for
  many NP-complete cases, 
  approximation algorithms.  
  In particular,
  for scoring rules we completely characterize the complexity of 
  weighted voter control.
  Our work shows
  that for quite a few important cases, either polynomial-time exact
  algorithms or polynomial-time approximation algorithms exist.

\end{abstract}

\section{Introduction}
In many real-world election systems the voters come with weights.
Examples range from stockholder elections weighted by shares, to the
US Electoral College, to the often-used example of the Nassau County Board
of Supervisors, to (in effect) any parliamentary system in which the 
parties typically vote as blocks, to Sweden's system of
wealth-weighted voting instituted in 1866 (and no longer used) where
``the wealthiest members of the rural communities received as many as
5,000 votes'' and ``in 10 percent of the districts the weighted votes
of just three voters could be
decisive''~\cite{con:b-chapter:sweden}.

So it is not surprising that in the study of manipulative attacks on
elections, weighted voting has been given great attention.  For
bribery and manipulation, two of the three most studied types of
manipulative attacks on elections, study of the case of weighted
voters has been extensively conducted.  Yet for the remaining one of
the three most studied types of attacks on elections, so-called
control attacks, almost no attention has been given to the case of
weighted voting; to the best of our knowledge, the only time this
issue has been previously raised is in two
M.S./Ph.D. theses~\cite{rus:t:borda,lin:thesis:elections}.  This lack
of attention is troubling, since the key types of control attacks,
such as adding and deleting voters, certainly do occur in many
weighted elections.

We study the complexity in weighted elections of
arguably the most important types of control---adding and deleting
voters---for various election systems.  We focus on
scoring rules, families of scoring rules, and 
(weak)Condorcet-consistent rules.
Control by 
deleting (adding)
voters asks
whether in a given election a given candidate can be made to win by
deleting (adding) at most a certain number of the voters (at most a
certain number of the members of the pool of potential additional
voters).  These control types model issues that are found in many
electoral settings, ranging from human to electronic.  They are
(abstractions of) issues often faced by people seeking to steer an
election, such as experts doing campaign management, and deciding for
example which $k$ people to offer rides to the polls.

Control was introduced (without
weights) in 1992 in the seminal paper by Bartholdi, Tovey, and
Trick~\cite{bar-tov-tri:j:control}.  Control has been the subject
of much attention since.  That attention, and the present paper, are
part of the line of work, started by Bartholdi, Orlin, Tovey, and Trick~\cite{bar-tov-tri:j:manipulating,bar-oli:j:polsci:strategic-voting,bar-tov-tri:j:control}, 
that seeks to determine for which types of manipulative attacks on
elections the attacker's task requires just polynomial-time
computation.
For a more detailed discussion of this line of work, 
we point the reader to the related work section
at the end of the paper
and to the 
surveys~\cite{fal-hem-hem-rot:b:richer,fal-hem-hem:j:cacm-survey,bra-con-end:b:comsoc}.

Our main results are as follows (see Section~\ref{ss:conclusions} for
tables summarizing our results).  First, in
Section~\ref{ss:scoring-protocols} we provide a detailed study of the
complexity of voter control under scoring protocols, for the case of 
fixed numbers of candidates.  We show that both constructive control
by adding voters and constructive control by deleting voters are in
$\p$ for $t$-approval
(and so this 
also covers 
plurality and 
$t'$-veto\footnote{If the number of candidates is fixed, then $t$-veto
  can be expressed as $(m-t)$-approval, where $m$ is the
  number of candidates. If the number of candidates is unbounded,
  then $t$-veto is not $t'$-approval.}) and are
$\np$-complete otherwise.  It is interesting to compare this result to
an analogous theorem regarding weighted coalitional manipulation:
There are cases where the complexities of voter control and
manipulation are the same (e.g., for plurality or for Borda) but there
are also cases where voter control is easier ($t$-approval for $t \geq
2$, for elections with more than $t$ candidates).  Is it ever possible
that weighted voter control is harder than weighted voting
manipulation? We show that weighted voter control is $\np$-hard
for (weak)Condorcet-consistent rules with at least three candidates.
Since weighted coalitional manipulation for the 3-candidate Llull system
is in $\p$~\cite{fal-hem-sch:c:copeland-ties-matter}, together with
the fact that Llull is weakCondorcet-consistent, this implies that
there is a setting where weighted voter control is harder than
weighted coalitional manipulation.

In Sections~\ref{ss:t-approval} and~\ref{ss:approx} we focus on the
complexity of weighted voter control under $t$-approval and $t$-veto, 
for the case of unbounded numbers of candidates.  
At the start of Section~\ref{ss:t-approval}, we will 
explain why these are the most interesting cases.
In
Section~\ref{ss:t-approval} we resolve six problems left open by
Lin~\cite{lin:thesis:elections}.  We establish the complexity of
weighted control by adding voters for $2$-approval, $2$-veto, and
$3$-approval, and of weighted control by deleting voters for
$2$-approval, $2$-veto, and $3$-veto.  In Section~\ref{ss:approx},
we give polynomial-time approximation algorithms for weighted voter
control under $t$-approval and $t$-veto.  Our algorithms seek to
minimize the number of voters that are added or deleted.

We believe that the complexity of weighted voter control,
and more generally 
the complexity of attacks on weighted elections, is an
important and 
interesting research direction that deserves much
further study. In particular, our research suggests
that it is worthwhile to seek 
$f(\cdot)$-approximation results for weighted
elections problems and that doing so can lead to interesting
algorithms.

\section{Preliminaries}
We assume that the reader is familiar with the basic notions of 
computational complexity theory
and the theory of algorithms.  Below we provide
relevant definitions and conventions regarding elections, election
rules, and control in elections. We also review some $\np$-complete
problems 
that we use in our reductions.

\subsubsection*{Elections}
We take an election to be a pair $E = (C,V)$, where $C$ is a set of
candidates and $V$ is a collection of voters.  
Each voter has a preference order over the set $C$. A \emph{preference
  order} is a total, linear order that ranks the candidates from the
most preferred one to the least preferred one. For example, if $C =
\{a,b,c\}$ and some voter likes $a$ best, then $b$, and then $c$, then
his or her preference order is $a > b > c$.
In weighted elections, each voter $v$ also has a positive integer
weight $\omega(v)$. A voter of weight $\omega(v)$ is treated by the
election system as $\omega(v)$ unweighted voters.  Given two
collections of voters, $V$ and $W$, we write $V+W$ to denote their
concatenation. 

\subsubsection*{Election Rules}
An election rule (or voting rule) is a function $R$ that given an
election $E = (C,V)$ returns a subset $R(E) \subseteq C$, namely those
candidates that are said to win the election. 

An $m$-candidate scoring rule is defined through a nonincreasing
vector $\alpha = (\alpha_1, \ldots, \alpha_m)$ of nonnegative
integers. For each voter $v$, each candidate $c$ receives
$\alpha_{\pos(v,c)}$ points, where $\pos(v,c)$ is the position of $c$
in $v$'s preference order. The candidates with the maximum total score are
the winners.  Given an election $E$ and a voting rule $R$ that assigns
scores to the candidates, we write $\score_E(c)$ to denote $c$'s total
score in $E$ under $R$. The voting rule used will always be clear from
context.  Many election rules are defined through families of scoring
rules, with one scoring vector for each possible number of
candidates. For example:
\begin{enumerate}
\item Plurality rule uses vectors of the form $(1,0, \ldots, 0)$.
\item $t$-approval uses vectors $(\alpha_1, \ldots, \alpha_m)$, where
  $\alpha_i = 1$ for each $i \in \{1, \ldots, t\}$, and $\alpha_i = 0$
  for $i > t$. By $t$-veto we mean the system that for $m$
  candidates uses the $(m-t)$-approval scoring vector.  For
  $m$-candidate $t$-approval and $t$-veto systems we will often treat
  each vote as a 0/1 $m$-dimensional approval vector that indicates
  which candidates receive points from the vote. Naturally, such a
  vector contains exactly $t$ ones for $t$-approval and exactly $t$
  zeroes for $t$-veto.\footnote{We emphasize that such a view of
    $t$-approval and $t$-veto is correct in settings where the set
    of candidates remains fixed. If the set of candidates were to
    change (e.g., as in control by adding/deleting candidates), then
    we would have to use the standard, preference-order-based
    definition.}

\item Borda's rule uses vectors of the form $(m-1, m-2, \ldots, 0)$,
  where $m$ is the number of candidates.
\end{enumerate}

Given an election $E = (C,V)$, a candidate $c$ is a (weak) Condorcet
winner if for every other candidate $d \in C -\{c\}$ it holds that
more than half (at least half) of the voters prefer $c$ to $d$. Note
that it is possible that there is no (weak) Condorcet winner in a
given election. We say that a rule $R$ is Condorcet-consistent if 
whenever there is a Condorcet winner he or she is the sole 
winner elected under $R$.
Analogously, a rule is weakCondorcet-consistent if it elects exactly
the weak Condorcet winners whenever they exist.
Every weakCondorcet-consistent system is 
Condorcet-consistent, but the converse does not always hold.

There are many 
Condorcet-consistent 
rules. We will briefly
touch upon the Copeland family of rules and the Maximin rule. For a given
election $E = (C,V)$ and two distinct candidates $c, d \in C$, we let
$N_E(c,d)$ be the number of voters that prefer $c$ to $d$.  Let
$\alpha$ be a rational number, $0 \leq \alpha \leq 1$. Under
Copeland$^\alpha$ the score of candidate $c \in C$ is defined as:
\[
\| \{ d \in C-\{c\} \mid N_E(c,d) > N_E(d,c) \} \| + \alpha \| \{ d
\in C-\{c\} \mid N_E(c,d) = N_E(d,c) \} \|,
\] and under Maximin the score of candidate $c\in C$ 
is defined as $\min_{d \in C-\{c\}}N_E(c,d)$. The
candidates with the highest score are winners.  Llull is another name
for Copeland$^1$. Clearly, Llull and Maximin are weakCondorcet-consistent.

\subsubsection*{Electoral Control}
We focus on constructive control by adding/deleting voters in weighted
elections. However, there are also other standard types of control
studied in the literature (e.g., control by adding/deleting
candidates and various forms of partitioning of candidates and voters; we
point the reader to Section~\ref{sec:related-work} for a discussion of
related work).

\begin{definition}\label{def:wcc}
  Let $R$ be a voting rule.  In both weighted constructive control by
  adding voters under rule $R$ ($R$-WCCAV) and weighted
  constructive control by deleting voters under rule $R$ 
  ($R$-WCCDV), our input contains a set of candidates $C$, a
  collection 
  of weighted voters $V$ 
(sometimes referred to as the registered voters)
with preferences over $C$, a
  preferred candidate $p \in C$, and a nonnegative integer $k$. In
  $R$-WCCAV we also have an additional collection $W$ 
  of weighted voters
  (sometimes referred to as the unregistered voters)
  with preferences over $C$. In these problems we ask the following
  questions:
  \begin{enumerate}
  \item {$R$}-WCCAV: Is there a subcollection $W'\!$ of $W$, of
    at most $k$ voters, such that $p \in R(C,\,V{+}W')$?
  \item {$R$}-WCCDV: Is there a subcollection $V'\!$ of $V$, of
    at most $k$ voters, such that $p \in R(C,\,V{-}V')$?
  \end{enumerate}
\end{definition}

Although in this paper we focus primarily 
on constructive control, Section~\ref{ss:scoring-protocols} makes some 
comments about the
so-called destructive variants of control problems. 
Given a voting rule $R$, weighted destructive control by
adding voters under rule $R$ ($R$-WDCAV) and weighted destructive
control by deleting voters under rule $R$ ($R$-WDCDV) are defined
analogously to their constructive variants, with the only difference
being that the goal is to ensure that the distinguished candidate $p$
is not a winner.

Note that in the above definitions the parameter $k$ defines the number
of voters that can be added/deleted, and not the total weight of the
voters that can be added/deleted. This is a standard approach when
modeling strategic behavior in weighted elections. For example,
in the study
of ``$R$-weighted-bribery''~\cite{fal-hem-hem:j:bribery},
bribing each weighted voter has unit cost
regardless of the voter's 
weight.

We will consider approximation algorithms for WCCAV and WCCDV under
$t$-approval and $t$-veto.  When doing so, we will assume that input
instances do not contain the integer $k$.  Rather, the goal is simply to
find (when success is possible at all) as small as possible a
collection of voters to add/delete such that $p$ is a winner of the
resulting election.  For a positive integer $h$, 
an $h$-approximation
algorithm for WCCAV/WCCDV is an algorithm that (when success 
is possible at all) always finds a solution
that adds/deletes at most $h$ times as many voters as an optimal
action does.
The notion of an 
$f(\cdot)$-approximation
algorithm for WCCAV/WCCDV is defined analogously, where the argument 
to $f$ is some variable related to the problem or instance.  
And the 
meaning of $\bigoh(f(\cdot))$-approximation 
algorithms will be similarly clear from context.
It is natural to worry about how the above seemingly 
incomplete definitions interact 
with the possibility that success might be impossible regardless of 
how many votes one adds/deletes.
However, for $t$-approval WCCDV and $t$-veto WCCDV
(and indeed, for any scoring rule), it is always possible to ensure that
$p$ is a winner, for example 
by deleting all the voters.
For $t$-approval WCCAV and $t$-veto WCCAV, it is possible to ensure
$p$'s victory through adding voters if and only if $p$ is a winner
after we add all the unregistered 
voters that approve of $p$. These observations
make it particularly easy to discuss and study approximation algorithms
for $t$-approval and for $t$-veto, because we can always easily check
whether there is some solution. 
For voting rules 
that don't have this easy-checking property,
such an analysis might be much more complicated.
(The reader may wish to compare our work with 
Brelsford et
al.'s
attempt at
framing a general election-problem approximation 
framework~\cite{bre-fal-hem-sch-sch:c:approximating-elections}.)

In this paper we do not consider candidate-control cases (such as
weighted constructive control by adding candidates and weighted
constructive control by deleting candidates, WCCAC and WCCDC).  The
reason is that for a 
bounded
number of candidates, when
winner
determination in the given weighted election
system is in $\p$ it holds that 
both WCCAC and WCCDC are in $\p$ by brute-force search. On the other
hand, if
the number of candidates is not bounded then candidate control is
already $\np$-hard for plurality (and $t$-approval and $t$-veto, in both
the constructive setting and the destructive setting) even without
weights~\cite{bar-tov-tri:j:control,hem-hem-rot:j:destructive-control,elk-fal-sli:j:cloning,lin:thesis:elections}.
Furthermore, many results for candidate control under Condorcet-consistent
rules can be claimed in the weighted setting.  For example, for
the Maximin rule and for the Copeland family of rules, hardness results
translate immediately, and it is straightforward to see that the existing
polynomial-time algorithms for the unweighted cases also work for the
weighted cases~\cite{fal-hem-hem:j:multimode}.

\subsubsection*{Weighted Coalitional Manipulation}
One of our goals is
to compare the complexity of weighted voter control with the
complexity of weighted coalitional manipulation (WCM)\@.
WCM is similar to WCCAV in that we also add voters, but it differs in
that (a)~we have to add exactly a given number of voters, and (b)~we
can pick the preference orders of the added voters.  It is quite
interesting to see how the differences in these problems' definitions
affect their complexities.

\begin{definition}
  Let $R$ be a voting rule. In $R$-WCM we are
given a weighted election $(C,V)$, a preferred candidate $p \in C$,
and a sequence $k_1, \ldots, k_n$ of positive integers. We ask whether it
is possible to construct a collection $W = (w_1, \ldots, w_n)$ of $n$
voters such that for each $i$, $1 \leq i \leq n$, $\omega(w_i) = k_i$,
and $p$ is a winner of the $R$ election $(C,\,V{+}W)$. 
The voters in $W\!$ are called manipulators.
\end{definition}

\subsubsection*{Computational Complexity}
In
our $\np$-hardness proofs we use reductions from the following
$\np$-complete problems.

\begin{definition}
  An instance of Partition consists of a sequence $(k_1, \ldots, k_t)$
  of positive integers whose sum is even.
 We ask whether there is a set $I \subseteq \{1, \ldots, t\}$
  such that $\sum_{i \in I}k_i = \frac{1}{2}\sum_{i=1}^t k_i$.
\end{definition}

In the proof of Theorem~\ref{t:avdv-scoring-protocols} we will use the
following restricted version of Partition, where we have greater
control over the numbers involved in the problem.

\begin{definition}
  An instance of Partition$\,'\!$ consists of a sequence $(k_1, \ldots,
  k_t)$ of positive integers, whose sum is even,
 such that (a)~$t$ is an even number, and
  (b)~for each $k_i$, $1 \leq i \leq t$, it holds that $k_i \geq
  \frac{1}{t+1}\sum_{j=1}^t k_j$.  We ask whether there is a set $I
  \subseteq \{1, \ldots, t\}$ of cardinality $\frac{t}{2}$ such that
  $\sum_{i \in I}k_i = \frac{1}{2}\sum_{i=1}^t k_i$.
\end{definition}

Showing the $\np$-completeness of this problem is a standard exercise.
(In particular, the $\np$-completeness of a variant of this 
problem 
is established as~\cite[Lemma~2.3]{fal-hem-hem:j:bribery}; the same
approach can be used to show the $\np$-completeness of 
Partition$'$.)
Our remaining hardness proofs are based on reductions from a 
restricted version of the well-known Exact-Cover-By-3-Sets problem.
This restricted version is still NP-complete~\cite{gar-joh:b:int}.

\begin{definition}
\label{def:x3c}
  An instance of X3C$\,'\!$ consists of a set $B =
  \{b_1, \ldots, b_{3t}\}$ and a family $\calS =
\{S_1,\allowbreak \ldots, S_n\}$
  of $3$-element subsets of $B$ such that
every element of $B$ occurs
in at least one and in at most three sets in ${\cal S}$.
We ask 
whether ${\cal S}$ contains an exact
cover for $B$, i.e., whether there exist
$t$ sets in  ${\cal S}$ whose union is $B$.
\end{definition}

\section{Results}
We now present our results. In
Section~\ref{ss:scoring-protocols} we focus on 
fixed numbers of
candidates in scoring protocols 
and (weak)Condorcet-consistent
rules. Then in Sections~\ref{ss:t-approval} and~\ref{ss:approx} we
consider case of an unbounded number of candidates, for $t$-approval and
$t$-veto.

\subsection{Scoring Protocols and Manipulation Versus  Control}\label{ss:scoring-protocols}

It is well-known that weighted manipulation of scoring protocols is
always hard, unless the scoring protocol is in effect
plurality or triviality~\cite{hem-hem:j:dichotomy}.  In contrast,
weighted voter control is easy for $m$-candidate $t$-approval.

\begin{theorem}
\label{t:t-approval-in-P}
For all $m$ and $t$, WCCAV and WCCDV for $m$-candidate $t$-approval are in $\p$.
\end{theorem}

\begin{proof}
Let $(C, V, W, p, k)$ be an instance
of WCCAV for $m$-candidate $t$-approval.
We can assume that we add only voters who approve of $p$.
We can also assume that we add the heaviest voters with a particular 
set of approvals,
i.e., if we add $\ell$ voters approving $p, c_1, \ldots, c_{t - 1}$,
we can assume that we added the $\ell$ heaviest voters approving
$p, c_1, \ldots, c_{t - 1}$.  Since there are only
${m-1 \choose t-1}$---which is a constant---different sets 
of approvals to consider, it
suffices to try all sequences of nonnegative integers
$k_1, \ldots, k_{m-1 \choose t-1}$ whose sum is at most $k$,
and for each such sequence 
to check whether adding the heaviest $k_i$ voters of the
$i$th approval collection
makes $p$ a winner.

The same argument works for WCCDV\@.  Here, we delete only voters that do
not approve of $p$, and again we delete the heaviest voters for each 
approval collection.
\end{proof}

One might think that the argument above works for any scoring
protocol, but this is not the case.  For example, consider the
3-candidate Borda instance where $V$ consists of one weight-1 voter $b
> p > a$ and $W$ consists of a weight-2 and a weight-1 voter with
preference order $a > p > b$.  Then adding the weight-1 voter makes
$p$ a winner, but adding the weight-2 voter does not.  And, in fact,
we have the following result.\footnote{We mention in passing that
  an analogue of this theorem in the model in
  which we are bounding the total weight of votes that can be
  added/deleted was obtained by Russell~\cite{rus:t:borda}.}

\begin{theorem}\label{t:bounded-borda}
WCCAV and WCCDV for Borda are $\np$-complete.  This result holds
even when restricted to a fixed number $m \geq 3$ of candidates.
\end{theorem}
\begin{proof}
We reduce from Partition.  Given a sequence $k_1, \ldots, k_t$ of
positive integers that sum to $2K$, construct an election with one
registered voter of weight $K$ voting $b > p > a > \cdots$, and 
$t$ unregistered voters with weights $k_1, \ldots, k_t$ voting $a > p > b > \cdots$.
Set the addition limit to $t$. 
It is easy to see that 
for $p$ to become a winner, $b$'s score (relative to $p$) needs to go down 
by at least $K$, while $a$'s score (relative to $p$) should not go up by
more than $K$.  It follows that $k_1, \ldots, k_t$  has a partition
if and only if $p$ can be made a winner.

We use the same construction for the deleting voters case.  Now, all
voters are registered and the deletion limit is $t$. Since we can't
delete all voters, and since our goal is to make $p$ a winner, we
can't delete the one voter voting $b > p > a > \cdots$.  The rest of the argument
is identical to the adding voters case.
\end{proof}

Interestingly, it is possible to extend the above proof to work for
all scoring protocols other than $t$-approval (the main idea stays the
same, but the technical details are more involved). And so, regarding
the complexity of WCCAV and WCCDV for scoring protocols with a fixed
number of candidates, the cases of Theorem~\ref{t:t-approval-in-P} are
the only P cases (assuming P $\neq$ NP).

\begin{theorem}\label{t:avdv-scoring-protocols}
  For each scoring protocol $(\alpha_1, \ldots, \alpha_m)$, if there
  exists an $i$, $1 < i < m$, such that $\alpha_1 > \alpha_i >
  \alpha_m$, then WCCAV and WCCDV for $(\alpha_1, \ldots, \alpha_m)$
  are $\np$-complete.
\end{theorem}
\begin{proof}
  Let $\alpha = (\alpha_1, \ldots, \alpha_m)$ be a scoring protocol
  such that there is an $i$ such that $\alpha_1 > \alpha_i >
  \alpha_m$.  Let $x$ be the third largest value in the set
  $\{\alpha_1, \ldots, \alpha_m\}$. We will show that WCCAV and WCCDV
  are $\np$-complete for scoring protocol $\beta = (\beta_1, \ldots,
  \beta_m) = (\alpha_1 - x, \ldots, \alpha_m-x)$.  While formally we
  have defined scoring protocols to contain only nonnegative values,
  using $\beta$ simplifies our construction and does not affect the
  correctness of the proof.
  To further simplify notation, given some candidates $x_1, \ldots,
  x_\ell$, by $F[ x_1 = \beta_{i_1}, x_2 = \beta_{i_2}, \ldots, x_\ell
  = \beta_{i_\ell}]$ we mean a fixed preference order that ensures,
  under $\beta$, that each $x_j$, $1 \leq j \leq \ell$, is ranked at a
  position that gives $\beta_{i_j}$ points. (The candidates not
  mentioned in the $F[\ldots]$ notation are ranked arbitrarily.)
  We let $\gamma_1$, $\gamma_2$, and $\gamma_3$ be the three highest
  values in the set $\{ \beta_1, \ldots, \beta_m\}$.  Clearly,
  $\beta_1 = \gamma_1 > \gamma_2 > \gamma_3 = 0$.
   
  We give a reduction from Partition to $\beta$-WCCAV (the membership
  of $\beta$-WCCAV in $\np$ is clear); let $(k_1, \ldots, k_t)$ be an
  instance of Partition, i.e., a sequence of positive integers that
  sum to $2K$.  We form an election $E = (C,V)$ where $C =
  \{p,a,b,c_4, \ldots, c_m\}$ and where the collection $V$ contains
  the following three groups of voters (for the WCCAV part of the proof
  below, we set $T=1$; for the WCCDV part of the proof we will use the
  same construction but with a larger value of $T$):
  \begin{enumerate}
  \item A group of $T$ voters, each with weight $K$ and preference order $F[b =
    \gamma_1, a = \gamma_2, p = 0]$.
  \item A group of $T$ voters, each with weight $K$ and preference order $F[p =
    \gamma_1, b = \gamma_2, a = 0]$.
  \item For each $c_i \in C$, there are $6$ collections of $2T$ voters, one collection for each
    permutation $(x,y,z)$ of $(p,a,b)$; the voters in each collection have weight $K$ and
    preference order $F[x = \beta_1, y = \beta_2, z = \beta_3, c_i =
    \beta_m]$.
  \end{enumerate}

  Let $M$ be the number of points that each of $a$, $b$, and $p$
  receive from the third group of voters (each of these candidates
  receives the same number of points from these voters). For each $c_i
  \in C$ and each $x \in \{p,a,b\}$, $x$ receives at least
  $4TK\gamma_1$ points more than $c_i$ from the voters in the third
  group (in each vote in the third group, $x$ receives at least as
  many points as $c_i$, and there are two collections of $2T$ voters
  where $x$ receives $\gamma_1$ points and $c_i$ receives $\gamma_m
  \leq 0$ points). Thus it holds that our candidates have the
  following scores:
  \begin{enumerate}
  \item $p$ has $M + TK\gamma_1$ points,
  \item $a$ has $M + TK\gamma_2$ points, 
  \item $b$ has $M + TK(\gamma_1+\gamma_2)$ points, and
  \item each candidate $c_i \in C$ has at most $M - 2TK\gamma_1$
    points.
  \end{enumerate}
  As a result, $b$ is the unique winner.  There are $t$ unregistered
  voters with weights $Tk_1, \ldots, Tk_t$, each with preference order
  $F = [a = \gamma_1, p = \gamma_2, b = 0]$. We set the addition limit
  to be $t$.  It is clear that irrespective of which voters are added,
  none of the candidates in $\{c_4, \ldots, c_m\}$ becomes a
  winner.

  If there is a subcollection of $(k_1, \ldots, k_t)$ that sums to
  $K$, then adding corresponding unregistered voters to the election
  ensures that all three of $p$, $a$, and $b$ are winners.  On the
  other hand, assume that there are unregistered voters of total
  weight $TL$, whose addition to the election ensures that $p$ is
  among the winners. For $p$ to have score at least as high as $b$, we must have
  that $L \geq K$. However, for $a$ not to have score higher than $p$,
  it must be that $L \leq K$. This means that $L = K$.  Thus it is
  possible to ensure that $p$ is a winner of the election by adding at
  most $t$ unregistered voters if and only if there is a subcollection
  of $(k_1, \ldots, k_t)$ that sums to $K$.  
  And, completing the proof, we note that 
  the reduction can be carried out in
  polynomial time.

  Let us now move on to the case of WCCDV\@. We will use the same
  construction, but with the following modifications:
  \begin{enumerate}
  \item Our reduction is now from Partition$'$. Thus without loss of 
    generality we can
    assume that $t$ is an even number and that for each $i$, $1 \leq i
    \leq t$, it holds that $k_i \geq \frac{1}{1+t}2K$.
  \item We set $T = \left \lceil
    \frac{t}{2}(t+1)\frac{\gamma_1}{\gamma_1-\gamma_2}\right\rceil+1$ (the
    reasons for this choice of $T$ will become apparent in the course
    of the proof; intuitively it is convenient to think of $T$ as of a
    large value that, nonetheless, is polynomially bounded with
    respect to $t$).
  \item We include the unregistered voters as ``the fourth group of
    voters.''
  \item We set the deletion limit to $\frac{t}{2}$.
  \end{enumerate}

  By the same reasoning as in the WCCAV case, it is easy to see that
  if there is a size-$\frac{t}{2}$ subcollection of $k_1, \ldots,
  k_t$ that sums to $K$, then deleting the corresponding voters
  ensures that $p$ is among the winners (together with $a$ and $b$).
  We now show that if there is a way to delete up to $\frac{t}{2}$
  voters to ensure that $p$ is among the winners, then the deleted
  voters must come from the fourth group, must have total weight $K$,
  and there must be exactly $\frac{t}{2}$ of them.  For the sake of
  contradiction, let us assume that it is possible to ensure $p$'s
  victory by deleting up to $\frac{t}{2}$ voters, of whom fewer than
  $\frac{t}{2}$ come from the fourth group. Let $s$ be the number of
  deleted voters from the fourth group ($s < \frac{t}{2}$) and let
  $x$ be a real number such that $xTK$ is their total weight. We have
  that $xTK$ is at most:
  \[ 
  xTK \leq 2TK - \frac{t - s}{1+t}(2TK) \leq 2TK\left(1 -
    \frac{\frac{t}{2}+1}{1+t}\right) = TK\frac{t}{1+t}.
  \]
  That is, we have $0 \leq x \leq \frac{t}{1+t}$.  Prior to deleting
  any voters, $a$ has $TK(\gamma_1-\gamma_2)$ points more than
  $p$. After deleting the $s$ voters from the fourth group, this
  difference decreases to $TK(1-x)(\gamma_1 - \gamma_2)$. If we
  additionally delete up to $\frac{t}{2}$ voters from the first three
  groups of voters, each with weight $K$, then the difference between the scores
  of $a$ and $p$ decreases, at most, to the following value (note that in
  each deleted vote both $a$ and $p$ are ranked at positions where
  they receive $\gamma_1$, $\gamma_2$ or $0$ points):
  \[
  TK(1-x)(\gamma_1 - \gamma_2) -\frac{t}{2}K\gamma_1 \geq
  TK\frac{1}{t+1}(\gamma_1-\gamma_2) - \frac{t}{2}K\gamma_1 =
  K\left(\frac{T(\gamma_1-\gamma_2)}{t+1} -
    \frac{\frac{t}{2}(t+1)\gamma_1}{t+1}\right) > 0.
  \]
  The final inequality follows by our choice of $T$. The above
  calculation shows that if there is a way to ensure $p$'s victory by
  deleting up to $\frac{t}{2}$ voters then it requires deleting
  exactly $\frac{t}{2}$ voters from the fourth group. The same
  reasoning as in the case of WCCAV shows that these $\frac{t}{2}$
  deleted voters must correspond to a size-$\frac{t}{2}$ subcollection of
  $(k_1, \ldots, k_t)$ that sums to $K$.
\end{proof}
As a side comment, we mention that WDCAV and
WDCDV for scoring protocols (that is, the destructive variants of
WCCAV and WCCDV) have simple polynomial-time algorithms: It suffices
to loop through all candidates $c$, $c \neq p$, and greedily
add/delete voters to boost the score of $c$ relative to $p$ as much as
possible.

Combining Theorems~\ref{t:t-approval-in-P}
and~\ref{t:avdv-scoring-protocols}, 
we obtain the following corollary, which
we contrast with an analogous result for
WCM~\cite{hem-hem:j:dichotomy}.

\begin{corollary}
  For each scoring protocol $(\alpha_1, \ldots, \alpha_m)$ the
  problems WCCAV
  and WCCDV
  are $\np$-complete if $\| \{\alpha_1, \ldots, \alpha_m\}\| \geq 3$
  and are in $\p$ otherwise.
\end{corollary}

\begin{theorem}[Hemaspaandra and Hemaspaandra~\cite{hem-hem:j:dichotomy}]
  For each scoring protocols $(\alpha_1, \ldots, \alpha_m)$, $m \geq
  2$, WCM is $\np$-complete if $\alpha_2 > \alpha_m$ and is in 
  $\p\!$ otherwise.
\end{theorem}

We see that for scoring protocols with a fixed number $m$ of
candidates, either WCM is harder than WCCAV and WCCDV (for the case of
$t$-approval with $2 \leq t < m$), or the complexity of WCM, WCCAV, and
WCCDV is the same ($\p$-membership for plurality and triviality, and
$\np$-completeness for the remaining cases).  For other voting rules,
it is also possible that WCM is easier than WCCAV and WCCDV\@.

\begin{theorem}\label{t:bounded-weakcondorcet}
For every weakCondorcet-consistent election system and for every
Condorcet-consistent election system,
WCCAV
and WCCDV are $\np$-hard.  This result holds
even when restricted to a fixed number $m \geq 3$ of candidates.
\end{theorem}

\begin{proof}
To show that WCCAV is $\np$-hard, we reduce from Partition.
Given a sequence $k_1, \ldots, k_t$ of
positive integers that sum to $2K$, construct an election with
two registered voters, one voter with weight 1 voting $p > a > b > \cdots$ and
one voter with weight $2K$ voting $b > p > a > \cdots$, and 
$t$ unregistered voters with weights $2k_1, \ldots, 2k_t$ voting
$a > p > b > \cdots$.  Set the addition limit to $t$.
Suppose we add unregistered voters to the election with a total
vote weight equal to $2L$.
\begin{itemize}
\item If $L < K$, then $b$ is the Condorcet winner, and thus the unique winner
of the election.
\item If $L > K$, then $a$ is the Condorcet winner, and thus the unique winner
of the election.
\item If $L = K$, then $p$ is the Condorcet winner, and thus the unique winner
of the election.
\end{itemize}
The WCCDV case uses the same construction.  Now, all
voters are registered and the deletion limit is $t$. 
Since we can delete at most $t$ of our $t+2$ voters, and
since our goal is to make $p$ a winner, we
can't delete the sole voter voting $b > p > a$, since then $a$ would
be the Condorcet winner.  The rest of the argument is similar to the
adding voters case.
\end{proof}

Let Condorcet be the election system whose winner set is exactly the set
of Condorcet winners.
Let weakCondorcet be the election system whose winner set is exactly the set
of weak Condorcet winners.

\begin{corollary}
For Condorcet and weakCondorcet, 
WCM is in $\p\!$ and WCCAV and WCCDV are $\np$-complete. 
This result holds even when restricted to a fixed number $m \geq 3$
of candidates.
\end{corollary}

\begin{proof}
It is immediate that WCM for Condorcet and weakCondorcet are in P.
To see if we have a ``yes''-instance of WCM, it suffices 
to check whether 
letting all the manipulators rank $p$ (the preferred candidate)
first and ranking all the remaining candidates in some arbitrary
order ensures $p$'s victory. $\np$-completeness of WCCAV and WCCDV
follows directly from Theorem~\ref{t:bounded-weakcondorcet}.
\end{proof}

Condorcet and weakCondorcet do not always have winners.  For those
who prefer their voting systems to always have at least one winner, we
note that WCM for $3$-candidate Llull is in
P~\cite{fal-hem-sch:c:copeland-ties-matter}.
\begin{corollary}\label{cor:llull}
For 3-candidate Llull, WCM is in $\p\!$ and WCCAV and WCCDV are $\np$-complete.
\end{corollary}

The main results of this section are also presented in Table~\ref{tab:bounded}
of Section~\ref{ss:conclusions}.

\subsection{\boldmath{\large$t$}-Approval and \boldmath{\large$t$}-Veto with an Unbounded Number of Candidates}\label{ss:t-approval}

Let us now look at the cases of $t$-approval and $t$-veto
rules, for an unbounded number of candidates.
The reason we focus on these
is that 
these are the most interesting families of scoring protocols 
whose complexity has not already been resolved
in the previous section.
The reason we say that is that Theorem~\ref{t:avdv-scoring-protocols}
shows that whenever we have at least three distinct values in a 
scoring vector, we have NP-completeness.  And so any family
that at even one number of candidates has three distinct values in 
its scoring vector is NP-hard for 
WCCAV and WCCDV\@.  Thus the really interesting 
cases are indeed $t$-approval and $t$-veto.

Our starting
point here is the work of Lin~\cite{lin:thesis:elections}, which
showed that for $t \geq 4$, WCCAV for $t$-approval and WCCDV for
$t$-veto are $\np$-complete, and that for $t \geq 3$, WCCDV for
$t$-approval and WCCAV for $t$-veto are $\np$-complete.  These results
hold even for the unweighted case.  It is also known that the
remaining unweighted cases are in
P~\cite{bar-tov-tri:j:control,lin:thesis:elections} and that WCCAV and
WCCDV for plurality and veto are in P~\cite{lin:thesis:elections}.
In this section, we look at and solve the remaining open cases,
WCCAV for $2$-approval, $3$-approval, and $2$-veto, and
WCCDV for $2$-approval,  $2$-veto, and $3$-veto. We start by
showing that $2$-approval-WCCAV is in $\p$.

\begin{theorem}
\label{t:WCCAV-2approval-in-P}
WCCAV for 2-approval is in $\p$.
\end{theorem}

\begin{proof}
We claim that Algorithm~\ref{alg:2approval-wccav} solves
$2$-approval-WCCAV in polynomial time. (In this algorithm and the proof of 
correctness, whenever we speak of the 
$r$ heaviest voters in voter set $X$, we mean the
$\min(r,\|X\|)$ heaviest voters in $X$.)
\begin{algorithm}[t]
\SetKw{KwAnd}{and}
\SetKw{KwReject}{reject}
\SetKw{KwAccept}{accept}
\SetCommentSty{textit}
\dontprintsemicolon
\SetSideCommentRight

   \small
   \SetAlCapFnt{\small}

  \KwIn{$(C, V, W, p, k)$}
  \ForAll{$c \in C -\{p\}$}
    {let $s_c = score_{(C,V)}(c) - score_{(C,V)}(p)$.}
  Delete from $W$ all voters that do not approve of $p$.\\
  \Repeat{no more changes.}
     {\ForAll{$c \in C -\{p\}$}
       {\lIf{the sum of the weights of the $k$ heaviest voters in $W$ that do not approve of $c$ is less than $s_c$}
         {\KwReject \\ \tcp{It is impossible to get $score(c) \leq score(p)$ by adding 
less than or equal to $k$
voters from $W$.}}}
     \ForAll{$c \in C -\{p\}$ \KwAnd $\ell \in \{1,\ldots, k-1\}$}
       {\If{the sum of the weights of the $k-\ell$ heaviest voters
           in $W$ that do not approve of $c$ is less than $s_c$}
         {delete from $W$ all voters approving $c$ except for the $\ell-1$ heaviest such voters.
\\           \tcp{We need to add at least $k-\ell+1$ voters that do not approve of $c$,
                 and so we can add at most $\ell-1$ voters approving $c$.}}}}
   \lIf{$\|W\| \geq k$}{\KwAccept}  \tcp{We can make $p$ a winner by adding the $k$ heaviest voters from $W$.}
   \If{$\|W\| < k$}
       {\lIf{adding all of $W$ will make $p$ a winner}{\KwAccept}\lElse{\KwReject}}
\caption{\label{alg:2approval-wccav}$2$-approval-WCCAV}
\end{algorithm}
It is easy to see that we never reject incorrectly in the repeat-until,
assuming that we don't incorrectly delete voters from $W$.
It is also easy to see that if we add $r$ voters approving $\{p,c\}$, we may
assume that we add the $r$ heaviest voters approving $\{p,c\}$ (this is also 
crucial in the proof of Theorem~\ref{t:t-approval-in-P}),
and so we never delete voters incorrectly in the
second for loop in the repeat-until.

If we get through the repeat-until without rejecting,
and we have fewer than $k$ voters
left in $W$, then adding all of $W$ is the best we can do (since all
voters in $W$ approve $p$).

Finally, if we get through the repeat-until, and we have at least $k$ voters
left in $W$, then adding the $k$ heaviest voters from $W$ will make $p$
a winner.
Why?  Let $c$ be a candidate in $C - \{p\}$.  Let $r$ be the number
of voters from $W$ that are added and that approve of $c$.
Since we made it through the repeat-until, we
know that
[the sum of the weights of the $k$ heaviest voters in $W$ that do not approve
 of $c$] is at least $s_c$.
We will show that after adding the voters, 
$score(c) - score(p) \leq 0$, which implies that $p$ is a winner.
If $r = 0$, $score(c) - score(p)$ =
$s_c$ - [the sum of the weights of the $k$ heaviest voters in $W$] $\leq 0$.
If $r > 0$, %
then [the sum of the weights of the $k-r$ heaviest voters
in $W$ that do not approve of $c$]  is at least $s_c$ (for otherwise
we would have at most $r - 1$ voters approving $c$ left in $W$).
And so $score(c) - score(p)$ = $s_c$ - [the sum of the weights of the $k-r$
heaviest voters in $W$ that do not approve of $c$] $\leq 0$.
\end{proof}

\begin{theorem}\label{t:2-veto}
WCCDV for 2-veto is in $\p$.
\end{theorem}

Instead of proving this theorem directly, we show a more general
relation between the complexity of $t$-approval/$t$-veto WCCAV
and WCCDV.

\begin{theorem}\label{thm:reduction}
  For each fixed $t$, it holds that $t$-veto-WCCDV
  ($t$-approval-WCCDV) polynomial-time many-one reduces to $t$-approval-WCCAV
  ($t$-veto-WCCAV).  
\end{theorem}
\begin{proof}
  We first give a reduction from $t$-veto-WCCDV to $t$-approval-WCCAV\@.
  The idea is that deleting a $t$-veto vote $v$ from $t$-veto election
  $(C,V)$ is equivalent, in terms of net effect on the scores, to
  adding a $t$-approval vote $v'$ to this election, where $v'$
  approves exactly of the $t$ candidates that $v$ disapproves of.  The
  problem with this approach is that we are to reduce $t$-veto-WCCDV
  to $t$-\emph{approval}-WCCAV and thus we have to show how to
  implement $t$-veto scores with $t$-approval votes.

  Let $(C,V,p,k)$ be an instance of $t$-veto-WCCDV, where $V =
  (v_1, \ldots, v_n)$. Let $m = \|C\|$. Let $\omega_{\max}$ be the
  highest weight of a vote in $V$.  We set $D$ to be a set of up to
  $t-1$ new candidates, such that $\|C\|+\|D\|$ 
  is a multiple of $t$.  We set $V_0$ to be a collection of
  $\frac{\|C\|+\|D\|}{t}$ $t$-approval votes,
  where each vote has weight $\omega_{\max}$ and each candidate in
  $C \cup D$ is approved in exactly one of the votes.
  For each vote $v_i$ in $V$ we create a set
  $C_i = \{c_i^1, \ldots, c_i^{(t-1)(m-t)}\}$ of candidates and we create a collection
  of voters $V_i = (v_i^1, \ldots, v^{m-t}_i)$. Each voter $v_i^j$, $1
  \leq j \leq m-t$, has weight $\omega(v_i)$ and approves of the
  $j$th candidate approved by $v$ and of the $t-1$ candidates 
  $c_i^{(j-1)(t-1)+1}, \ldots,
  c_i^{j(t-1)}$.

  We form an election $E' = (C',V')$, where $C' = C \cup D \cup
  \bigcup_{i=1}^{n}C_i$ and $V' = V_0 + V_1 + \cdots + V_n$.
  For each candidate $c$, let
  $s_c$ be $c$'s $t$-veto score in $(C,V)$; it is easy to see that
  $c$'s $t$-approval score in $E'$ is
  $\omega_{\max}+s_c$.
  Furthermore, each candidate $c \in C' - C$ has $t$-approval
  score at most $\omega_{\max}$ in $E'$.

  We form an instance $(C',V',W,p,k)$ of $t$-approval-WCCAV, where $W
  = (w_1, \ldots, w_n)$, and for each $i$, $1 \leq i \leq n$,
  $\omega(w_i) = \omega(v_i)$, and $w_i$ approves exactly of those
  candidates that $v_i$ disapproves of. It is easy to see that adding
  voter $w_i$ to $t$-approval election $(C',V')$ has the same net
  effect on the scores of the candidates in $C$ as does deleting $v_i$
  from $t$-veto election $(C,V)$.\medskip

  Let us now give a reduction from $t$-approval-WCCDV to
  $t$-veto-WCCAV\@.  The idea is the same as in the previous reduction;
  the main part of the proof is to show how to implement $t$-approval
  scores with $t$-veto votes.  Let $(C,V,p,k)$ be an instance of
  $t$-approval-WCCDV, where $V = (v_1, \ldots, v_n)$. Let $m = \|C\|$
  and let $\omega_{\max}$ be the highest weight of a vote in $V$. We
  set $D$ to be a set of candidates such that $t \leq \|D\| \leq 2t-1$ and $\|C\| +
  \|D\| = s\cdot t$ for some integer $s$, $s \geq 3$ (note that for
  our setting to not be trivial it must be that $m > t$).
  We set $V_0$ to be a collection of $4n(s-2)$ votes, each with weight
  $\omega_{\max}$; each candidate from $C$ is approved in all these
  votes whereas each candidate from $D$ is disapproved in at least
  half of them (since $t \leq \|D\| \leq 2t-1$, it is easy to
  construct such votes).  For each vote $v_i$ in $V$, we create
  a collection $V_i$ of $(s-1)$ votes satisfying the following
  requirements: (a)~each candidate approved in $v_i$ is also approved
  in each of the votes in $V_i$, and (b)~each candidate not approved
  in $v_i$, is approved in exactly $(s-2)$ votes in $V_i$.  (Such
  votes are easy to construct: We always place the top $t$ candidates
  from $v_i$ in the top $t$ positions of the vote; for the remaining
  positions, in the first vote we place the candidates in some
  arbitrary, easily computable order, and in each following vote we
  shift these candidates cyclically by $t$ positions with respect to
  the previous vote.)  Each vote in $V_i$ has weight $\omega(v_i)$.

  We form an election $E' = (C', V')$, where $C' = C \cup D$ and $V' =
  V_0 + V_1 + \cdots + V_n$. For each candidate $c$, let
  $s_c$ be $c$'s $t$-approval score in $(C,V)$; it is easy to see that
  $c$'s $t$-veto score in $E'$ is $4n(s-2)\omega_{\max} + (s-2)(\sum_{i=1}^n\omega(v_i)) +
  s_c$. Furthermore, each candidate from $D$ has $t$-veto score at most
  $3n(s-2)\omega_{\max}$ in $E'$.

  We form an instance $(C',V', W, p, k)$ of $t$-veto-WCCAV, where $W =
  (w_1, \ldots, w_n)$, and for each $i$, $1 \leq i \leq n$,
  $\omega(w_i) = \omega(v_i)$, and $w_i$ disapproves of exactly those
  candidates that $v_i$ approves of. It is easy to see that adding
  voter $w_i$ to $t$-veto election $(C',V')$ has the same net effect
  on the scores of candidates in $C$ as deleting voter $v_i$ from
$t$-approval election
  $(C,V)$ has.  Furthermore, since each candidate in $D$ has at least
  $n\omega_{\max}$ fewer points than each candidate in $C$, the fact
  that adding $w_i$ increases scores of candidates in $D$ does not
  affect the correctness of our reduction\@.
\end{proof}

All other remaining cases
(WCCDV for 2-approval, WCCAV for 3-approval, WCCAV
for 2-veto, and WCCDV for 3-veto) are $\np$-complete.   
Interestingly, in contrast
to many other $\np$-complete weighted election problems, 
we need only
a very limited set of weights to make the reductions work.  

\begin{theorem}\label{t:hardness}
WCCAV for 2-veto and 3-approval and
WCCDV for 2-approval and 3-veto are $\np$-complete.
\end{theorem}

\begin{proof}
Membership in NP is immediate, so it suffices to prove NP-hardness.
We will first give the proof for WCCDV for 2-approval.
By Theorem~\ref{thm:reduction}
this also immediately gives the result for WCCAV for 
2-veto.
We will reduce from X3C$'$ from Definition~\ref{def:x3c}.
Let $B = \{b_1, ..., b_{3t}\}$ and let
${\cal S} = \{S_1, ..., S_n\}$ be
a family of 3-element subsets of $B$ such that
every element of $B$ occurs
in at least one and in at most three sets in ${\cal S}$.
We construct the following
instance $(C,V,p,k)$ of WCCDV for 2-approval. We set
$C = \{p\} \cup \{b_j \ | \ 1 \leq j \leq 3t\}
\cup \{s_i, s'_i \ | \ 1 \leq i \leq n\}
\cup \{d_0, d_1, \ldots, d_{3t}\}$
($d_0, d_1, \ldots, d_{3t}$ are dummy candidates that are used for padding).
For $1 \leq j \leq 3t$, 
let $\ell_j$ be the number of sets in ${\cal S}$ that contain $b_j$.
Note that $1 \leq \ell_j \leq 3$.
$V$ consists of the following voters:

\begin{center}
\begin{tabular}{ccl}
weight & preference order& \\
2 &  $s_i > s_{i\phantom{_1}}' > \cdots$ &     \multirow{4}{*}{$\left.\rule{0cm}{0.9cm}\right\}$ for all $1 \leq i \leq n$ and $S_i = \{b_{i_1}, b_{i_2}, b_{i_3}\}$}\\
1 &  $s_i > b_{i_1} > \cdots$\\
1 &  $s_i > b_{i_2} > \cdots$\\
1 &  $s'_i > b_{i_3} > \cdots$\\
2 &  $p\phantom{_i} > d_{0} > \cdots$\\
$3 - \ell_j$ &
$b_j > d_j > \cdots$ & for all $1 \leq j \leq 3t$ such that
$\ell_j < 3$.
\end{tabular}
\end{center}

Note that $\score(s_i) = 4$, $\score(s'_i) = 3$, $\score(b_j) = 3$,
$\score(p) = 2$, and $\score(d_j) \leq 2$.  We set $k = n + 2t$ and
we claim that ${\cal S}$ contains an exact cover if and only if
$p$ can become a winner after deleting at most $n + 2t$ voters.

$(\Rightarrow)$:
Delete the $(n-t)$ weight-2 voters corresponding to the sets not in the cover
and delete the $3t$ weight-1 voters corresponding to the sets in the cover.
Then the score of $p$ does not change, the score of each $s_i$ decreases by 2,
the score of each $s'_i$ decreases by at least 1,
and the score of each $b_j$ decreases by 1.
So, $p$ is a winner.

$(\Leftarrow)$:  We need to delete $3t$ voters to decrease the
score of every $b_j$ voter by 1.  After deleting these $3t$ voters, there
are at most $t$ values of $i$, $1 \leq i \leq n$, such that
the score of $s_i$ and the score of $s'_i$ are at most 2.

If there are exactly $t$ values of $i$, $1 \leq i \leq n$, such that
the score of $s_i$ and the score of $s'_i$ are at most 2, then these
$t$ values of $i$ correspond to a cover.  If there are less than
$t$ values of $i$, $1 \leq i \leq n$, such that
the score of $s_i$ and the score of $s'_i$ are at most 2, then 
the remaining voters that are deleted, and there are at most
$n-t$ of them, need to decrease the score of $s_i$ and/or $s'_i$ for
more than $n-t$ values of $i$, $1 \leq i \leq n$.  But that is not possible,
since there is no voter that approves of both $s_i$ or $s'_i$ and 
$s_{j}$ or $s'_j$ for $i \neq j$.

Note that this construction uses only weights 1 and 2. 
In fact, we can establish NP-completeness for WCCDV for 2-approval 
for every set of allowed weights of size at least two
(note that if the set of weights has size one, the problem
is in P, since this is in essence the unweighted
case~\cite{lin:thesis:elections}).
Since the reductions of Theorem~\ref{thm:reduction} do not change the
set of voter weights, we have the same result for
WCCAV for 2-veto.

So, suppose
our weight set contains $w_1$ and $w_2$, $w_2 > w_1 > 0$.
We modify the construction above as follows.
We keep the same set of candidates and we change the voters as follows.

\begin{center}
\begin{tabular}{cccl}
\# & weight & preference order& \\
1 &  $w_2$ &  $s_i > s_{i\phantom{_1}}' > \cdots$ & \multirow{4}{*}{$\left.\rule{0cm}{0.9cm}\right\}$ for all $1 \leq i \leq n$ and $S_i = \{b_{i_1}, b_{i_2}, b_{i_3}\}$}\\
1 & $w_1$ &  $s_i > b_{i_1} > \cdots$\\
1 & $w_1$ &  $s_i > b_{i_2} > \cdots$\\
1 & $w_1$ &  $s'_i > b_{i_3} > \cdots$\\
2 & $w_1$ &  $p\phantom{_i'} > d_0 > \cdots$ & if $w_2 \leq 2w_1$\\
1 & $w_2$ & $p\phantom{_i'} > d_0 > \cdots$ & if $w_2 > 2w_1$\\
$\ell - \ell_j$ & $w_1$ &  $b_j > d_j > \cdots$ & for all $1 \leq j \leq 3t$.
\end{tabular}
\end{center}

Here, $\ell$ is the smallest integer such that $\ell w_1 > \max(2w_1, w_2)$.
Note that $\ell \geq 3$ and so $\ell - \ell_j$ is never negative.
Note that $\score(s_i) = w_2 + 2w_1$, $\score(s'_i) = w_2 + w_1$,
$\score(b_j) = \ell w_1$, $\score(p) = \max(2w_1,w_2)$,
and $\score(d_j) \leq \max(2w_1,w_2)$.
The same argument as above shows that
${\cal S}$ contains an exact cover if and only if
$p$ can become a winner after deleting at most $n + 2t$ voters.

We now turn to the proof for WCCDV for 3-veto.   Our construction
will use only weights 1 and 3.
Since the reductions of Theorem~\ref{thm:reduction} do not change the
set of voter weights, weights 1 and~3 also suffice to get
NP-completeness for WCCAV for 3-approval.
Given the instance of X3C$'$ described above, we construct the following
instance $(C,V,p,k)$ of WCCDV for 3-veto. We set
$C = \{p\} \cup B \cup \{s_i \ | \ 1 \leq i \leq n\} \cup \{r,d, d' \}$
($d$ and $d'$ are dummy candidates that are
used for padding)
and $V$ consists of the following voters:

\begin{center}
\begin{tabular}{cccl}
\# & weight & preference order& \\
1 & 3 &  $\cdots > p > s_i\phantom{'} > r\phantom{_{i_1}}$ &  \multirow{4}{*}{$\left.\rule{0cm}{0.9cm}\right\}$ for all $1 \leq i \leq n$ and $S_i = \{b_{i_1}, b_{i_2}, b_{i_3}\}$}\\
1 & 1 &  $\cdots > p > s_i\phantom{'} > b_{i_1}$\\
1 & 1 &  $\cdots > p > s_i\phantom{'} > b_{i_2}$\\
1 & 1 &  $\cdots > p > s_i\phantom{'} > b_{i_3}$\\
$3n-3t$ & 1 & $\cdots > d > d'\phantom{_i} > r\phantom{_{i_1}}$\\
$3n-3$ & 1 & $\cdots > d > d'\phantom{_i} > s_{i\phantom{_1}}$ & for all $1 \leq i \leq n$\\
$3n+1-\ell_j$ & 1 & $\cdots > d > d'\phantom{_i} > b_{j\phantom{_1}}$ & for all $1 \leq j \leq 3t$.
\end{tabular}
\end{center}

It is more convenient to count the number of vetoes for each candidate than to count
the number of approvals.   Note that
$\vetoes(s_i) = 3n+3$, $\vetoes(b_j) = 3n+1$, $\vetoes(r) = 6n-3t$, 
$\vetoes(p) = 6n$, and $\vetoes(d) = \vetoes(d') \geq 3n$.
We claim that ${\cal S}$ contains an exact cover if and only if
$p$ can become a winner (i.e., have a lowest number of vetoes)
after deleting at most $n + 2t$ voters.

$(\Rightarrow)$:
Delete the $(n-t)$ weight-3 voters corresponding to the sets not in the cover
and delete the $3t$ weight-1 voters that veto $p$ and that
correspond to the sets in the cover.  
Then $\vetoes(s_i) = \vetoes(b_j) = \vetoes(r) = \vetoes(p) = 3n$ and 
$\vetoes(d) = \vetoes(d') \geq 3n$.  
So, $p$ is a winner.

$(\Leftarrow)$:
We can assume that we delete only voters that veto $p$. 
Suppose we delete $k_1$ weight-1 voters and $k_2$ weight-3 voters,
$k_1 + k_2 \leq n + 2t$. 
After this deletion, $\vetoes(p) = 6n - k_1 - 3k_2$,
$\vetoes(r) = 6n - 3t - 3k_2$, and $\vetoes(b_j) \leq 3n+1$.
In order for $p$ to be a winner, we need
$\vetoes(p) \leq \vetoes(r)$.  This implies that $k_1 \geq 3t$.
We also need $\vetoes(p) - \vetoes(b_j) \leq 0$.  Since
$\vetoes(p) - \vetoes(b_j) \geq  6n - k_1 - 3k_2 - (3n + 1) \geq
6n - (n+2t-k_2) - 3k_2 - 3n - 1  = 2n - 2t - 2k_2 - 1$, it follows
that $k_2 \geq n-t$.
So we delete $3t$ weight-1 votes and $n-t$ weight-3 votes, and after
deleting these voters $\vetoes(p) = 3n$.  In order for $p$ to
be a winner, we can delete
at most one veto for each $b_j$ and at most three vetoes for each $s_i$.  
This implies that the set of deleted weight-1 voters corresponds to 
a cover.   
\end{proof}

\subsection{Approximation and Greedy Algorithms}\label{ss:approx}
When problems are computationally difficult, such as being
NP-complete, it is natural to wonder whether good polynomial-time
approximation algorithms exist.  So, motivated by the NP-completeness
results discussed earlier in this paper for most cases of WCCAV/WCCDV
for $t$-approval and $t$-veto, this section studies greedy and other
approximation algorithms for those problems.  (Recall that WCCAV is
NP-complete for $t$-approval, $t \geq 3$, and for $t$-veto, $t \geq
2$, and WCCDV is NP-complete for $t$-approval, $t \geq 2$, and for
$t$-veto, $t \geq 3$.)  Although we are primarily interested in
constructing good approximation algorithms, we are also interested in
cases where particular greedy strategies can be shown to fail to
provide good approximation algorithms, as doing so helps one eliminate
such approaches from consideration and sheds light on the 
approach's limits of applicability.
First, we will establish a connection to the weighted
multicover problem, and we will use it to obtain approximation results.
Then we will obtain an  approximation algorithm
that will work by direct action on our problem.  
Table~\ref{tab:approx}
in Section~\ref{ss:conclusions}
summarizes our results on approximation
algorithms for $t$-approval/$t$-veto WCCAV/WCCDV\@.\medskip

\subsubsection{A Weighted Multicover Approach}
Let us first consider the extent to which known algorithms for
the Set-Cover family of problems apply to our setting. Specifically,
we will use the following multicover problem.

\begin{definition}
  An instance of Weighted Multicover (WMC) consists of a set
  $B = \{b_1, \ldots, b_m\}$, a sequence $r = (r_1, \ldots, r_m)$ of
  nonnegative integers (covering requirements), a collection $\calS =
  (S_1, \ldots, S_n)$ of subsets of $B$, and a sequence $\omega =
  (\omega_1, \ldots, \omega_n)$ of positive integers (weights of the
  sets in $\calS$). The goal is to find a minimum-cardinality
  set $I \subseteq \{1,
  \ldots, n\}$ such that for
  each $b_j \in B$ it holds that 
$r_j \leq \sum\limits_{i\in I \land b_j \in S_i}\omega_i$, 
   or to declare that no such set exists.
\end{definition}
That is, given a WMC instance we seek a smallest collection of subsets
from $\calS$ that satisfies the covering requirements of the elements
of $B$ (keeping in mind that a set of weight $\omega$ covers each of
its elements $\omega$ times).  WMC is an extension of 
Set-Cover with unit costs.
We will not define here the problem known as 
covering integer programming (CIP)
(see~\cite{kol-you:j:cip}).  However, that problem will be 
quite important to us here.  The reason is that we observe
that 
WMC is a special case of 
CIP
(with multiplicity constraints but) without packing constraints;
footnote~\ref{footnote:ky} below
is in effect describing how to embed our problem in 
that problem.
An approximation algorithm of 
Kolliopoulos and Young~\cite{kol-you:j:cip} for 
CIP
(with multiplicity constraints but) without packing constraints,
applied to the special case of WMC, gives the following 
result.\footnote{\label{footnote:ky}%
This follows from the sentence
starting ``Our second algorithm finds a solution'' on 
page~496 of \cite{kol-you:j:cip} (which itself follows 
from their Theorem~8), 
keeping in mind that 
we have none of their so-called packing constraints, and so 
we may take it that what they call $\epsilon$ is one and the matrix
and vector they call $B$ and $b$ won't be a factor here.
Their vector $a$ corresponds to our $r_j$'s; the element in
the $j$th row and $i$th column of their matrix $A$ will for us be 
set to 
$\omega_i$ if $S_i$ contains $b_j$ and $0$ otherwise; 
we set their 
cost vector $c$ to be a vector of all $1$'s;
we set their 
multiplicity vector $d$ to be a vector of all $1$'s;
their vector $x$ corresponds to the characteristic function of our $I$;
and 
their $\alpha$ will be 
Theorem~\ref{thm:ky}'s bound $t$ on 
the number of elements of $B$ contained in any $S_i$.}

\begin{theorem}[Kolliopoulos and Young~\cite{kol-you:j:cip}]\label{thm:ky}
   There is a polynomial-time algorithm that when given an instance of WMC
   in which each set contains at most $t$ elements gives an
   $\bigoh(\log t)$-approximation.
\end{theorem}

For $t$-approval both WCCAV and WCCDV naturally translate
to equivalent WMC instances.
We consider WCCAV first. Let $(C,V,W, p, k)$ be an instance of
$t$-approval-WCCAV, where $W = (w_1, \ldots, w_n)$ is the collection
of voters that we may add. We assume without loss of generality 
that each voter in $W$
ranks $p$ among its top $t$ candidates (i.e., approves of $p$).  

We form an instance $(B,r,\calS,\omega)$ of WMC as follows. We set
$B = C - \{p\}$.
For each $c \in B$, we set its covering requirement to be $r_c =
\score_{(C,V)}(c) \ominus
\score_{(C,V)}(p)$, where $i \ominus j =_{def} \max(0, i-j)$. 
For each vote $w \in W$, let
$S_w$ be the set of candidates that $w$ does not approve of.
By our
assumption regarding each voter ranking $p$ among its top $t$
candidates, no $S_w$ contains $p$. We set $\calS = (S_{w_1},
\ldots, S_{w_n})$ and we set $\omega = (\omega(w_1), \ldots,
\omega(w_n))$.  It is easy to see that a set $I \subseteq \{1,
\ldots, n\}$ is a solution to this instance of WMC (that is, $I$
satisfies all covering requirements) if and only if 
adding the voters $\{w_i
\mid i \in I\}$ to the election $(C,V)$ ensures that $p$ is a
winner. The reason for this is the following: If we add voter $w_i$
to the election then for each candidate $c \in
S_{w_i}$, the difference between the score of $c$ and the score of $p$
decreases by $\omega(w_i)$, and for each candidate $c \not\in
S_{w_i}$ this difference does not change. The covering
requirements are set to guarantee that $p$'s score will match or
exceed the scores of all candidates in the election.

We stress that in the above construction we did not assume $t$ to be a
constant. Indeed, the construction applies to $t$-veto just as well as
to $t$-approval. 
So using Theorem~\ref{thm:ky} we obtain the following result.

\begin{theorem}\label{cor:wccav-log}
There is a polynomial-time $\bigoh(\log m)$-approximation
  algorithm for $t$-approval-WCCAV\@. 
There is a polynomial-time 
algorithm that when given an instance of $t$-veto-WCCAV
($t \in \mathbb{N}$) gives an 
$\bigoh(\log t)$-approximation.
\end{theorem}
\begin{proof}
  It suffices to use the reduction of $t$-approval/$t$-veto to WMC and
  apply the algorithm from Theorem~\ref{thm:ky}. For the case of
  $t$-approval, the reduction guarantees that each set in the WMC
  instance contains at most $m$ elements.  For the case of
  $t$-veto, each of these sets contains at most $t$ elements.
\end{proof}

We can obtain analogous results for the case of $t$-approval/$t$-veto
and WCCDV\@. One can either provide a direct reduction from these
problems to WMC or notice that the reductions given in the proof of
Theorem~\ref{thm:reduction} maintain approximation properties.

\begin{theorem}\label{cor:wccdv-log}
There is a polynomial-time 
algorithm that when given an instance of $t$-approval-WCCDV
($t \in \mathbb{N}$) gives an 
$\bigoh(\log t)$-approximation.
There
  is a polynomial-time 
$\bigoh(\log m)$-approximation algorithm for $t$-veto-WCCDV\@.
\end{theorem}

\subsubsection{A Direct Approach}
Using algorithms for WMC, we were able to obtain relatively strong
algorithms for WCCAV/WCCDV under $t$-approval and $t$-veto. However,
with this approach we did not find approximation algorithms for
$t$-approval-WCCAV and $t$-veto-WCCDV whose approximation ratios do
not depend on the size of the election. In the following we will seek
direct algorithms for these problems.

We now show that a very simple greedy approach yields a
polynomial-time $t$-approximation algorithm for $t$-approval-WCCAV
and $t$-veto-WCCDV\@.
(Recall that this means that in cases when making $p$ win is possible,
the number of voters our algorithm adds/deletes to reach victory is never more
than $t$ times that of the optimal set of additions/deletions.)

Let GBW (greedy by weight) define the following very simple
algorithm for WCCAV\@.  (The votes are the weighted $t$-approval vectors
induced by the preferences of the voters.)  
(Pre)discard all
unregistered votes that do not approve of the preferred candidate $p$.
Order the (remaining) unregistered votes from heaviest to lightest,
breaking ties in voter weights in some simple, transparent way (for
concreteness, let us say by lexicographic order on the votes'
representations).  GBW goes through the unregistered votes in that
order, and as it reaches each vote it adds the vote exactly if the
vote disapproves of at least one candidate whose score (i.e., total
weight of approvals) is currently strictly greater than that of $p$.
It stops successfully when $p$ has become a winner and unsuccessfully
if before that happens the algorithm runs out of votes to consider.
The following result says that GBW is a $t$-approximation algorithm
for $t$-approval-WCCAV, and also 
for $t$-veto-WCCDV, using the obvious analogue of GBW for $t$-veto-WCCDV, 
which 
we will also call GBW.\footnote{For completeness and clarity, we describe 
what we mean by GBW for   
$t$-veto-WCCDV\@.
Order all votes that do not approve of $p$ from heaviest to lightest,
breaking ties in voter weights in some simple, transparent way (for
concreteness, let us say by lexicographic order on the votes'
representations).  GBW goes through these votes in that
order, and as it reaches each vote it removes the vote exactly if the
vote approves of at least one candidate whose score (i.e., total
weight of approvals) is currently strictly greater than that of $p$.
It stops successfully when $p$ has become a winner and unsuccessfully
if before that happens the algorithm runs out of such votes to consider.}

\begin{theorem} \label{thm:gbw}
  Let $t \geq 3$.  The polynomial-time
  greedy algorithm GBW is a $t$-approximation algorithm for
  $t$-approval-WCCAV and $t$-veto-WCCDV\@;  
  and there are instances in which GBW's
  approximation factor on each of these problems is no better than~$t$.
\end{theorem}

We prove Theorem~\ref{thm:gbw}'s upper and lower bound parts 
separately, through the following two lemmas from which the theorem 
immediately follows.  

\begin{lemma}\label{t:GBW-t-deferred-av-dv}
  Let $t \geq 3$.  
  There are instances on which 
  the polynomial-time
  greedy algorithm GBW 
  has an approximation factor on 
  $t$-approval-WCCAV 
  no better than~$t$.
  There are instances on which
  the polynomial-time
  greedy algorithm GBW 
  has an approximation factor on 
  $t$-veto-WCCDV 
  no better than~$t$.
\end{lemma}

\begin{lemma} \label{lemma:upper}
  Let $t \geq 3$.  The polynomial-time
  greedy algorithm GBW is a $t$-approximation algorithm for
  $t$-approval-WCCAV and $t$-veto-WCCDV\@.
\end{lemma}

The 
proof
of our lower-bound claim, Lemma~\ref{t:GBW-t-deferred-av-dv},
consists of a somewhat 
detailed pair of constructions, and is of less interest 
than 
the upper-bound part of Theorem~\ref{thm:gbw}, 
namely Lemma~\ref{lemma:upper}.
We thus 
defer to the appendix the proof 
of Lemma~\ref{t:GBW-t-deferred-av-dv}.

\begin{proof}[{Proof of Lemma~\protect\ref{lemma:upper}}]

  Let us now prove the two claims that GBW is a $t$-approximation algorithm.
  We will prove the result for $t=3$ and WCCAV, 
  but it will be immediately clear
  that our proof straightforwardly generalizes to all greater~$t$;
  and the WCCDV case follows using Theorem~\ref{thm:reduction}.

  Clearly GBW is a polynomial-time algorithm.  
  Consider a given input instance of $t$-approval-WCCAV, with
  preferred candidate $p$.  Without loss of generality, 
  assume all unregistered voters
  approve of $p$.  We will say a candidate ``has a gap'' (under the
  current set of registered voters and whatever unregistered voters
  have already been added) if that candidate has strictly more weight
  of approvals than $p$ does.  For each candidate $d$ who has a gap,
  $d \neq p$, define $i_d$ to be the minimum number of unregistered
  voters one has to add to remove $d$'s gap; that is, if one went from
  heaviest to lightest among the unregistered voters, adding in turn
  each that disapproved of $d$, $i_d$ is the number of voters one would add
  before $d$ no longer had a gap.  If for any candidate $d$ it holds
  that no integer realizes $i_d$, then control is impossible using
  the unregistered voter set.  Clearly, any successful addition of
  voters must add at least $\max_d i_d$ voters (the max throughout
  this proof is over all candidates initially having a gap).

  Let us henceforth assume that control is possible in the input
  case.  We will show that 
  after having added at most 
  $3 \cdot \max_d i_d$ voters GBW will have
  made $p$ a winner, and so GBW is a 3-approximation algorithm.  
  By way of contradiction,
  suppose that after 
  $3 \cdot \max_d i_d$ additions some candidate, $z$, still has a gap.

  \emph{Case 1 [In at least $\max_d i_d$ of the first $3 \cdot \max_d
    i_d$ votes added by GBW, $z$ is not approved].}  Since for the
  last one of these to be added $z$ must still have had a gap before
  the addition, each earlier vote considered that disapproved $z$ had
  a gap for $z$ when it was considered and so would have been added
  when reached.  So, keeping in mind that $i_z \leq \max_d i_d$, we in
  fact must have added the $i_z$ heaviest voters disapproving of $z$,
  and so contrary to the assumption, $z$ no longer has a gap after
  these additions.
  
  \emph{Case 2 [Case 1 does not hold].}  So $z$ is approved in at
  least $1 + 2\cdot \max_d i_d$ of the added votes.  What made 
  the final one of the added votes, call it $v'$, 
  eligible for addition?  It must be
  that some candidate, say $y$, still had a gap just before 
  $v'$ was added.

  \emph{Case 2a [$y$ is disapproved in at least $\max_d i_d$ of the
    $2\cdot \max_d i_d$ votes added before $v'$ that approved $z$].}
  Then, since until $y$'s gap was removed no unregistered voters
  disapproving of $y$ would be excluded by GBW, $y$'s $i_y$ heaviest
  voters will have been added.  So contrary to Case $2$'s assumption,
  $y$ does not have a gap when we get to adding vote $v'$.

  \emph{Case 2b [Case 2 holds but Case 2a does not].}  Then $y$ is 
  approved in at least $1+\max_d i_d$ of the $2\cdot \max_d i_d$ votes before
  $v'$ that GBW added that approve $z$.  So we have $1+\max_d i_d$ votes
  added approving of exactly $z$ and $y$.  But then who made the 
  last of \emph{those} $1+\max_d i_d$ votes, call it $v''$, eligible
  to be added?  It must be that some candidate $w$ had a gap 
  up through $v''$.  
  But at the moment before adding $v''$ we would have added
  $\max_d i_d \geq i_w$ votes approving exactly $z$ and $y$ and 
  so disapproving $w$, and since $w$ allegedly still had a gap, 
  we while doing so under GBW would have in fact added the 
  $i_w$ heaviest voters disapproving of $w$, and so 
  $w$'s gap would have been removed before $v''$, so
  contrary to our assumption $w$ was not the gap that made $v''$ 
  eligible.
\end{proof}

One might naturally wonder how GBW performs on $t$-veto-WCCAV and
$t$-approval-WCCDV\@.  
By an argument far
easier than that used in the above proof of 
Lemma~\ref{lemma:upper}, 
in both of these cases
GBW provides a $t$-approximation algorithm.

\begin{theorem}\label{t:FORMERLY-no-constant} 
GBW is a $t$-approximation 
algorithm for t-veto-WCCAV\@.
GBW is a $t$-approximation 
algorithm for t-approval-WCCDV\@.
\end{theorem}

\begin{proof}
Consider $t$-veto-WCCAV\@.  
Let $p$ be the preferred candidate.
For each candidate 
$d$
with an initial positive 
``gap'' 
relative to the preferred candidate $p$
(i.e., a surplus over $p$ in total weight of approvals),
let $i_d$ be as defined in 
the proof of Lemma~\ref{lemma:upper}.
(Recall that 
$i_d$ is the number of votes we would need to add to remove 
the surplus of $d$ over $p$ 
if we took the unregistered votes, discarded all
that didn't simultaneously approve $p$ and disapprove $d$,
and then added those one at a time from heaviest to lightest until
the gap was removed.)  Clearly, $\sum i_d$, where the sum is taken over 
those candidates with an initial surplus relative to $p$, 
is an upper bound on the 
number of votes added by GBW.  This is 
true since GBW works by adding 
extra votes from heaviest to lightest, restricted to those vetoing 
a candidate who at that point has 
a positive gap relative to $p$;  so under GBW each gap will be closed by the 
largest weight votes that address it.
On the other hand, 
in any overall optimal solution
$i_d$ is a lower bound 
on the smallest number of votes from 
that solution's added-vote set that 
would suffice 
to remove $d$'s positive gap (since it takes $i_d$ even if we use the 
heaviest votes addressing the gap).
In the overall optimal solution 
each 
added vote 
narrows at most $t$ gaps.
So GBW's solution
uses at worst $t$ times as many added votes as does the optimal solution.

The claim for 
$t$-approval-WCCDV follows by 
Theorem~\ref{thm:reduction}.  
\end{proof}

This result replaces
a flawed claim in an earlier version of this paper that GBW 
and some of its cousins do not
provide $\bigoh(1)$ approximations for these problems.  
Of course, 
having a $t$-approximation for these two problems is not 
wildly exciting, since for these problems the multicover-based
approach from earlier in this section showed that for some 
function 
$f(t)$, with $f(t) = \bigoh(\log t)$, we even have 
$f(t)$-approximation
algorithms for these problems.  However, if the constant of
the ``big oh'' of that other algorithm is large, it is possible 
that for sufficiently small values of $t$ 
the above approach may give a better 
approximation.  Also, we feel that it is interesting to learn about 
the behavior of explicit heuristics, especially attractive 
approaches such as greedy algorithms.

\section{Related Work}
\label{sec:related-work}
The study of the complexity of (unweighted) electoral control was
initiated by Bartholdi, Tovey, and Trick~\cite{bar-tov-tri:j:control},
who considered constructive control by adding/deleting/partitioning
candidates/voters under the plurality rule and under the Condorcet
rule (that is, the rule that chooses Condorcet winner whenever there
is one, and has no winners otherwise).  The various types of control
model at least some of the flavor of actions that occur in the real
world, such as voter suppression and targeted get-out-the-vote drives
(see the survey of Faliszewski, Hemaspaandra, and
Hemaspaandra~\cite{fal-hem-hem:j:cacm-survey} for more examples and
discussions).  A major motivation for the study of control was to
obtain ``complexity barrier'' results, that is, results that show that
detecting opportunities for various control attacks is computationally
difficult.  In particular, Bartholdi, Tovey, and Trick focused on
$\np$-hardness as the measure of computational difficulty.

This research direction was continued by Hemaspaandra, Hemaspaandra,
and Rothe~\cite{hem-hem-rot:j:destructive-control}, who were 
the first to study
destructive control attacks on elections.
Since then, many authors have studied electoral control in
many varied settings and under many different rules; we refer the
reader to the survey~\cite{fal-hem-hem:j:cacm-survey}.
Some recent research, not covered in that survey, includes
complexity-of-control results for
the $t$-approval family of rules~\cite{lin:thesis:elections}, 
for Bucklin's rule (and for fallback, its extension for truncated
votes)~\cite{erd-rot:c:fallback,erd-pir-rot:c:voter-partition-in-bucklin-and-fallback-voting},
for maximin~\cite{fal-hem-hem:j:multimode}, 
for range voting~\cite{men:jtoappear:range-voting},
and
for Schultze's rule and the ranked pairs
rule~\cite{par-xia:strategic-schultze-ranked-pairs}.  
In the present paper, we compare control and manipulation.  
The recent paper~\cite{fit-hem-hem:ctoappear:control-manipulation}
studies settings in which both control and manipulation 
are occurring. 
Researchers have, in the quite different setting 
of electing members to fill a fixed-size, multimember panel,
defined variants of control 
that have coexisting 
constructive and destructive
aspects~\cite{mei-pro-ros-zoh:j:multiwinner}.  There is 
also work analyzing counting
variants of control~\cite{fal-woj:c:counting-control}, where the goal
is not only to decide if a given control attack is possible, but also
to count the number of ways in which this attack can be carried out.

The complexity-barrier research line turned out to be very
successful. For most voting rules that were considered, a significant
number of control attacks are $\np$-hard. Indeed, it is even possible
to construct an artificial election system resistant to all types of
control attacks~\cite{hem-hem-rot:j:hybrid}.  However, there are also
a number of results that suggest that in practice the complexity
barrier might not be as strong as one might at first think.  For
example, Faliszewski et
al.~\cite{fal-hem-hem-rot:j:single-peaked-preferences} and Brandt et
al.~\cite{bra-bri-hem-hem:c:sp2} have shown that if the votes are
restricted to being single-peaked, then many control problems that are
known to be $\np$-complete become polynomial-time solvable. Indeed,
this often holds even if elections are just nearly
single-peaked~\cite{fal-hem-hem:c:nearly-sp}, as many real-world
elections seem to be.  Similarly, some initial experimental results of
Rothe and Schend~\cite{rot-sch:empirical-control}---published very
recently---suggest that, at least under certain distributions and
settings, some $\np$-hard control problems can be solved in practice
on many instances.
As part of a different line of research,
Xia~\cite{xia:c:vote-operations} has studied the asymptotic behavior
of the number of voters that have to be added to/deleted from a
randomly constructed election in a successful control action.

There are a number of other problems that 
involving changing the structure of elections.
These problems include
candidate
cloning~\cite{elk-fal-sli:j:cloning,elk-fal-sli:c:decloning} (where it
is possible to replace a given candidate $c$ with a number of its
clones), or the possible winner problem when new alternatives
join~\cite{che-lan-mau-mon:c:possible-winners-adding,lan-mon-xia:c:new-alternatives-new-results}
(where some additional, not yet ranked candidates can be
introduced). This last problem is also related to the possible winner
problem with truncated ballots~\cite{bau-fal-lan-rot:c:lazy-voters}.

The only papers that directly raise the issue of weighted control are,
to the best of our knowledge, the theses of Russell~\cite{rus:t:borda}
and Lin~\cite{lin:thesis:elections}. However, we also mention the
papers of Baumeister et al.~\cite{bau-roo-rot-sch-xia:weighted-pw},
and of Perek et al.~\cite{per-fal-ros-pin:c:losing-voters}, where the
authors, in effect, consider problems of affecting the result of an
election through picking the weights of the voters. 
(The paper of Perek
et al.\ motivates its study differently, but in effect studies a
constrained variant of choosing voter weights.) Their problems are
similar to, though different from, simultaneous (multimode) addition
and deletion of voters~\cite{fal-hem-hem:j:multimode}.

This paper has given $f(\cdot)$-approximation 
results for weighted election 
control problems.
Elkind and
Faliszewski~\cite{elk-fal:c:shift-bribery} have 
given a 2-approximation 
algorithm for a weighted, bribery-related case.

\section{Conclusions}\label{ss:conclusions}

\begin{table}[tp]
\begin{center}
\setlength{\tabcolsep}{10pt}
\begin{tabular}{p{3.5cm}|cc|c}
  \multicolumn{1}{c|}{} & WCCAV & WCCDV & WCM\\
\hline&&&\\[-0.5em]
  Plurality & $\p$ (Thm.~\ref{t:t-approval-in-P}) & $\p$ (Thm.~\ref{t:t-approval-in-P})& $\p$~\cite{con-lan-san:j:when-hard-to-manipulate} \\[5pt]
  {$t$-approval, $2 \leq t < m$} & $\p$ (Thm.~\ref{t:t-approval-in-P}) & $\p$ (Thm.~\ref{t:t-approval-in-P}) & $\npcs$~\cite{hem-hem:j:dichotomy}\\[5pt]
  {Borda}          & $\npcs$~(Thm.~\ref{t:bounded-borda}) & $\npcs$~(Thm.~\ref{t:bounded-borda}) & $\npcs$~\cite{hem-hem:j:dichotomy}\\[8pt]
\raggedright 
$\alpha = (\alpha_1, \ldots, \alpha_m)$, $\|\{\alpha_1, \ldots, \alpha_m\}\| \geq 3$ & 
\multirow{2}*{$\npcs$~(Thm.~\ref{t:avdv-scoring-protocols})} & \multirow{2}{*}{$\npcs$~(Thm.~\ref{t:avdv-scoring-protocols})} & \multirow{2}*{$\npcs$~\cite{hem-hem:j:dichotomy}} \\[25pt]
  Llull (3 candidates) & $\npcs$~(Cor.~\ref{cor:llull}) & $\npcs$~(Cor.~\ref{cor:llull}) & $\p$~\cite{fal-hem-sch:c:copeland-ties-matter}\\[10pt]
\parbox[c]{3.5cm}{(weak)Condorcet-consistent rules}                    
& {$\np$-hard}  (Thm.~\ref{t:bounded-weakcondorcet})
& {$\np$-hard}  (Thm.~\ref{t:bounded-weakcondorcet})
& \parbox[c]{0.8in}{various \mbox{complexities}}%
\end{tabular}
\caption{\label{tab:bounded}Our results for the complexity of control by adding/deleting voters in
  weighted elections for 
any fixed number of candidates,
${m \geq 3}$, compared to the complexity of weighted coalitional manipulation. 
}
\end{center}
\end{table}

\begin{table}[tp]
\begin{center}
\setlength{\tabcolsep}{15pt}
\begin{tabular}{l|cc}
  \multicolumn{1}{c|}{} & WCCAV & WCCDV \\
\hline&&\\[-0.5em]
  $t$-approval && \\
  \quad $t = 2$ & $\p$ (Thm.~\ref{t:WCCAV-2approval-in-P}) & $\npc$ (Thm.~\ref{t:hardness}) \\
  \quad $t = 3$ & $\npc$ (Thm.~\ref{t:hardness}) & $\npc$~\cite{lin:thesis:elections}\\
  \quad $t \geq 4$ & $\npc$~\cite{lin:thesis:elections}& $\npc$~\cite{lin:thesis:elections} \\[10pt]
  $t$-veto && \\
  \quad $t = 2$ & $\npc$ (Thm.~\ref{t:hardness}) & $\p$ (Thm.~\ref{t:2-veto}) \\
  \quad $t = 3$ & $\npc$~\cite{lin:thesis:elections} & $\npc$ (Thm.~\ref{t:hardness}) \\
  \quad $t \geq 4$ & $\npc$~\cite{lin:thesis:elections} & $\npc$~\cite{lin:thesis:elections} 
\end{tabular}
\caption{\label{tab:unbounded}The complexity of control by adding and
  deleting voters for ${t}$-approval and
  ${t}$-veto with an unbounded number of candidates.}
\end{center}
\end{table}

\begin{table}[tp]
\begin{center}
\setlength{\tabcolsep}{10pt}
\begin{tabular}{c|c c}
 \multicolumn{1}{c|}{} & WCCAV & WCCDV \\
\cline{1-3}&&\\[-0.50em]
 \multirow{2}{*}{$t$-approval} & $\bigoh(\log m)$ (Thm.~\ref{cor:wccav-log}) & $\bigoh(\log t)$ (Thm.~\ref{cor:wccdv-log}) \\
                        & $t$ (Thm.~\ref{thm:gbw})                    & $t$ (Thm.~\ref{t:FORMERLY-no-constant}) \\[12pt]
 \multirow{2}{*}{$t$-veto} & $\bigoh(\log t)$ (Thm.~\ref{cor:wccav-log}) & $\bigoh(\log m)$ (Thm.~\ref{cor:wccdv-log})\\
                        & $t$ (Thm.~\ref{t:FORMERLY-no-constant})      & $t$ (Thm.~\ref{thm:gbw}) \\
\end{tabular}

\caption{\label{tab:approx}Approximation ratios of our algorithms for
  WCCAV and WCCDV under $t$-approval and $t$-veto.}
\end{center}
\end{table}

We have studied voter control under a number of voting rules,
including scoring protocols, families of scoring protocols,
and the 
(weak)Condorcet-consistent rules. 
We have shown that the complexity of voter control can be quite
different from the complexity of weighted coalitional manipulation:
there are natural voting rules for which weighted coalitional
manipulation is easy but weighted voter control is hard, and 
there are natural rules where the
opposite is the case. Furthermore, we have shown that for weighted voter
control under $t$-approval and $t$-veto, there are good, natural
approximation algorithms.
Our results for voter control in weighted elections are 
summarized in Tables~\ref{tab:bounded},~\ref{tab:unbounded},
and~\ref{tab:approx}.

\subsubsection*{Acknowledgments}   
We are very 
grateful to the anonymous \mbox{AAMAS 2013} referees for 
helpful comments and suggestions.
This work was supported 
in part 
by grants
AGH-\allowbreak{}11.11.120.865, 
NCN-DEC-2011/03/B/ST6/01393, 
NCN-UMO-2012/06/M/ST1/00358, 
and 
NSF-CCF-\{0915792,\allowbreak{}1101452,\allowbreak{}1101479\},
and two Bessel Awards from the 
Alexander von Humboldt Foundation.

\bibliographystyle{alpha}
\bibliography{gry-sp3ijcai}  %

\appendix

\section{Additional Details Related to
Section~\protect\ref{ss:approx}}
We present here the deferred proof 
of Lemma~\ref{t:GBW-t-deferred-av-dv}
and some other details related to 
Section~\protect\ref{ss:approx}.

\begin{proof}[{Proof of Lemma~\protect\ref{t:GBW-t-deferred-av-dv}}]
Our goal is to 
show that 
GBW sometimes really 
does use fully $t$ times the optimal number of added/deleted 
votes, for the cases 
in question.  
Examples are (somewhat detailed but)
not hard to construct, and the lower bound even holds for $t=2$, though
in Section~\ref{ss:t-approval} we obtained an exact solution by
a different approach.
However, one does have to be careful to set the ``gap'' pattern
created by the unregistered voters to be a realizable one.  
For our $t$-approval-WCCAV construction, this will be easy to do
directly.  
For our $t$-veto-WCCDV construction, 
we will establish realizability
through a small tool---which we hope may prove useful elsewhere---that lets one
set up certain patterns of gaps.  We state the tool below 
as Tool~\ref{tool:tool1}.

Fix any $t \in \{2,3,4,\ldots\}$.  We will now construct an instance
of $t$-approval-WCCAV on which GBW uses $t$ times as many additions as
the optimal strategy.  Our construction will have $2t$ candidates:
the preferred candidate $p$, candidates $a_1,\ldots,a_t$, and
candidates $d_1,\ldots,d_{t-1}$.  Now, suppose that under the
votes of the registered voters, the ``gaps'' are as follows.  For each
candidate $a_i$, the total weight of approvals of $a_i$ exceeds the
total weight of approvals of $p$ by exactly $2t$.  And for each
candidate $d_i$, the total weight of approvals of $d_i$ equals the
total weight of approvals of $p$.  
This can easily be realized, namely by our registered voter 
set being one weight-$2t$ voter
who approves of each $a_i$.

Our set of unregistered voters will be as follows.  There will be 
one unregistered voter, call it ``nice,'' of weight $2t$, who
approves of $p$ and each of the $t-1$ candidates $d_i$, and
disapproves of each of the $t$ candidates $a_i$.  For each $j$, $1\leq
j \leq t$, we will have a single unregistered 
voter, call it $\alpha_j$, of weight $3t$,
who approves of $p$ and of each $a_i$ other than $a_j$, and disapproves 
of $a_j$ and all the $d_i$'s.

Note that GBW will add all $t$ voters $\alpha_i$.  But ideal would be to
add the single voter called ``nice,'' since doing so suffices to 
make $p$ a winner.
So for each $t \geq 2$ we have constructed a setting
where GBW for $t$-approval-WCCAV takes
$t$ times more than the optimal number of added votes.
It also holds that for each $t \geq 2$, we can similarly
construct a setting
where GBW for
$t$-veto-WCCDV takes $t$ times more than the optimal number of deleted
votes, 
and can prove that setting to be realizable.
In fact, we can do so by following something of 
the flavor of  the above
scheme, except with a slightly different vote set that adjusts it to
handle the case of deleting voters, and with more 
care regarding realizability. Here is the construction.
Fix any $t \in \{2,3,4,\ldots\}$.  Our candidate set will again be 
the preferred candidate $p$, candidates $a_1,\ldots,a_t$, and
candidates $d_1,\ldots,d_{t-1}$.  Let us specify the voter set.
We will put into our voter set a collection of weight-1 
votes such that the gaps in total approval weight 
\emph{relative to $d_1$} created by those votes are as follows.
Each of $d_2$ through $d_{t-1}$ have the same total approval weight
as $d_1$.  The total approval weight of $p$ exceeds that of 
$d_1$ by $3t^2+3t$.  And the total approval weight of each $a_i$ exceeds
that of $d_1$ by $3t^2$.

As Tool~\ref{tool:tool1} below, we will observe that for 
$2t$-candidate $t$-approval voting, any gap pattern where the 
gaps are all multiples of $t$ can be realized.  
Since in the current proof we are using 
$2t$-candidate $t$-veto, and that is the same as
$2t$-candidate $t$-approval, 
Tool~\ref{tool:tool1} applies here.
In particular, 
Tool~\ref{tool:tool1} 
easily builds a set of weight-1
votes realizing precisely our desired set of gaps.
(The exact number of weight-1 votes used in 
this construction is not important.  However, from the gaps mentioned above
and the vote-set size mentioned in the tool, the precise number is 
easily seen to be 
$(3t+3+t(3t))(2t-1)$.)

We are not yet done building our voter set.  We will also 
have in our voter set
one voter, call it ``nice,'' of weight $2t$, who approves of exactly
all $t$ of the $a_i$'s. 
And for each $j$, $1\leq j \leq t$, we will
have one voter of weight $3t$ who approves of exactly $a_j$ and all
$t-1$ of the $d_i$'s.

Under the entire set of votes 
created above---the votes from the tool combined with ``nice'' and the 
other $t$ votes just mentioned---it is easy to see that $d_1$ is 
a candidate having the least total approval weight, and it is tied 
in total approval weight with each other $d_i$.  
The total approval weight of $p$ exceeds that of $d_1$ by $3t$.
And each $a_i$ exceeds $d_1$ in total approval weight by $5t$.

However, in light of the pattern of votes and weights we have here, 
it is clear that GBW (in its version for $t$-veto) will
delete the $t$ weight-$3t$ voters.  (Note that
the votes added by Tool~\ref{tool:tool1} are all weight-1 votes,
and so are highly unattractive to GBW\@.)
But ideal would be to
delete the single voter called ``nice,'' since doing so suffices to 
make $p$ a winner.
So for each $t \geq 2$ we have constructed a realizable setting
where GBW for $t$-veto-WCCDV takes
$t$ times more than the optimal number of deleted votes.
\end{proof}

Within the above proof, we referred 
to and used 
a small tool that can build certain patterns of vote weight
gaps in certain approval elections.  
It would be 
an overreach to claim that 
this is 
a McGarvey-like tool, since this is a different setting than, and is a far less
flexible result than, the famous theorem of
McGarvey~\cite{mcg:j:election-graph}.  
However, it in a small way is a
tool that perhaps might be useful elsewhere, and so we 
state and prove this
modest tool.

\begin{tool}\label{tool:tool1}
  Let $t\geq 2$.  Let $n_1,\ldots,n_{2t-1}$ be any list of
  nonnegative integers each divisible by $t$.  Then
  there exists a collection of $t$-approval
  votes, over $2t$ candidates, such that under those votes, relative
  to the candidate getting the fewest approvals, the list of gaps
  in number of approvals between that candidate and the other $2t-1$
  candidates is precisely $(n_1,\ldots,n_{2t-1})$.  
  Furthermore, this can be done with 
  $(2t-1)(\sum n_i)/t$ unweighted (i.e., weight 1) votes.  
  It alternatively can be done with 
  $(2t-1)^2$ weighted votes (or even $(2t-1)\|\{i~|~n_i \neq 0\}\|$ 
  weighted votes).
\end{tool}

\begin{proof}
Consider an election with $2t$ candidates, where the votes cast are
$t$-approval votes.  Consider the collection of $2t-1$ votes, each of
weight one, in which the votes all approve of a particular candidate
(for this example, let that one be the first candidate), and
the remaining $t-1$ approvals cyclically rotate around the other
candidates.  So the $t$-approval votes, viewed as bit vectors, are
these: $1\,1^{t-1}\, 0^t$, $1\,0\, 1^{t-1}\, 0^{t-1}$, $\ldots$,
$1\,0^t\, 1^{t-1}$, $1\,1\,0^t\, 1^{t-2}$, $\ldots$,
$1\,1^{t-1}\,0^t\,1$.  Note that the first candidate is approved in
all $2t-1$ of those votes, and each other candidate is approved in
exactly $t-1$ of those votes.  So this collection of votes sets
a gap of $t$ in favor of the first candidate, 
between the total approval weight of the first candidate and
that of each other candidate
And the difference in total approval weight 
between each other pair of candidates is zero.

Given a gap pattern as stated in the tool, where each gap above the 
least-approved candidate 
(call that candidate $c$) 
is a multiple of $t$, we can simply use the 
approach of the above paragraph repeatedly, to boost 
each other candidate, $d$, 
one at a time to whatever multiple of $t$ 
it is supposed to exceed $c$ by in total approval weight.
(In this, 
$d$ will play the role ``the first candidate'' did in
the previous paragraph.)
If $d$'s
surplus relative to $c$ is $kt$ and we wish 
to use only weight-1 votes, we can 
do this for $d$ with $k(2t-1)$ weight-1 votes.  Otherwise, we can do this 
for $d$ with $2t-1$ weight-$k$ votes.
So 
the total number of votes used is as 
given in the statement
of this tool.
\end{proof}

This appendix is not seeking to provide a comprehensive study of which
gap collections are realizable under $t$-approval voting, nor is it
seeking to find the smallest number of voters needed to realize 
realizable gap collection.  That is an interesting direction for study,
but is not our goal here.
However, we 
mention that there clearly exist some gap collections
that cannot be realized.  For example, 
the ``then there exists''
claim of Tool~\ref{tool:tool1} is not even always 
true if one 
removes the assumption of divisibility by $t$.  An example showing this 
is the following.  
Consider a 4-candidate setting where votes will be 
2-approval votes, and we desire a gap list relative to the least-approved
candidate of $(1,1,1)$, i.e., each of the other candidates has one more 
approval than does the least-approved candidate.  
Clearly, the total number of approvals of
any set of votes achieving this is $4B+3$, where $B$ is whatever number 
of approvals the least-approved candidate happens to get under the vote 
set one is trying, and so the total number of approvals is odd.  However,
any vote set of 2-approval votes has an even total number of approvals.
So this gap collection cannot be realized.

\end{document}